%% file: main.tex
\documentclass[english,11pt]{article}
\usepackage[T1]{fontenc}
\usepackage[latin9]{inputenc}
\usepackage{amsmath}
\usepackage{amsthm, amsfonts}
\usepackage{fullpage}
\usepackage[margin=1in]{geometry}
\usepackage{amssymb}
\usepackage{graphicx}
\usepackage{subcaption}
\usepackage{braket}
\usepackage{enumitem}
\usepackage{thmtools,thm-restate}
\usepackage{xcolor}
\usepackage{tikz}
\usetikzlibrary{quantikz}
\usepackage[bookmarks=false]{hyperref}
\definecolor{darkgreen}{rgb}{0,0.5,0}
\hypersetup{linktocpage=true,colorlinks=true,citecolor=red,linkcolor=darkgreen}

\usepackage[ruled,vlined,linesnumbered]{algorithm2e}

\newtheorem{theorem}{Theorem}[section]

\newtheorem{proposition}[theorem]{Proposition}
\newtheorem{lemma}[theorem]{Lemma}
\newtheorem{corollary}[theorem]{Corollary}
\newtheorem{conjecture}[theorem]{Conjecture}

\makeatletter

\newcommand{\expref}[2]{{\texorpdfstring{\hyperref[#2]{#1~\ref{#2}}}{#1~\ref{#2}}}} 
\newcommand{\secref}[1]{\expref{Section}{#1}}
\newcommand{\thmref}[1]{\expref{Theorem}{#1}}

\newcommand{\appref}[1]{\expref{Appendix}{#1}}
\newcommand{\lref}[1]{\expref{Lemma}{#1}}
\newcommand{\corref}[1]{\expref{Corollary}{#1}}
\newcommand{\conjref}[1]{\expref{Conjecture}{#1}}
\newcommand{\pref}[1]{\expref{Proposition}{#1}}
\newcommand{\figref}[1]{\expref{Figure}{#1}}

\theoremstyle{plain}

\theoremstyle{plain}

\theoremstyle{remark}

\usepackage{algorithmic}

\makeatother

\renewcommand{\inf}{\mathrm{Inf}}
\newcommand{\Inf}{\mathrm{Inf}}
\newcommand{\maxinf}{\mathrm{MaxInf}}

\newcommand{\poly}{\mathrm{poly}}
\newcommand{\ui}{{\boldsymbol{i}}}
\newcommand{\eps}{\epsilon}
\usepackage{babel}
\providecommand{\lemmaname}{Lemma}
\providecommand{\remarkname}{Remark}
\providecommand{\theoremname}{Theorem}
\usepackage{color}
\newcommand{\ba}{{{a}}}
\newcommand{\R}{{\mathbb R}}
\newcommand{\N}{\mathbb{N}}
\newcommand{\C}{\mathbb{C}}
\newcommand{\pmone}{\{\pm 1\}}
\newcommand{\cbnorm}[1]{\|#1\|_{\mathrm{cb}}}

\newcommand{\E}{\mathbb{E}}
\newcommand{\Ef}{\E f}
\newcommand{\BZ}{\mathbf{0}}

\newcommand{\Diag}{\mathrm{Diag}}
\newcommand{\beq}{\begin{equation}}
\newcommand{\eeq}{\end{equation}}

\newcommand{\Var}{\mathrm{Var}}
\newcommand{\CP}{\mathcal{P}}
\newcommand{\NC}{\mathcal{NC}}
\newcommand{\bi}{{\boldsymbol{i}}}
\newcommand{\bj}{{\boldsymbol{j}}}
\newcommand{\comment}[1]{{}}
\newcommand{\BU}{\boldsymbol{U}}
\newcommand{\U}{\boldsymbol{U}}
\newcommand{\BA}{\mathbf{A}}

\newcommand{\x}{\boldsymbol{x}}

\newcommand{\CA}{\mathcal A}
\newcommand{\bone}{\boldsymbol{1}}

\newcommand{\tr}{\mathrm{tr}}
\newcommand{\BR}{\R}

\newcommand{\ind}{\mathbf{1}}
\newcommand{\op}{\mathrm{op}}
\newcommand{\BE}{\E}
\newcommand{\fhat}{\widehat{f}}
\newcommand{\del}{\partial}

\newcommand{\cb}{\mathrm{cb}}
\newcommand{\hf}{\widehat{f}}
\newcommand{\bb}{\boldsymbol{b}}
\renewcommand{\phi}{\varphi}
\date{} 

\begin{document}

\title{Influence in Completely Bounded Block-multilinear Forms and Classical Simulation of Quantum Algorithms}
\author{ Nikhil Bansal\thanks{\texttt{bansaln@umich.edu}. Supported in part by the NWO VICI grant 639.023.812.}\\ \small University of Michigan \and
 Makrand Sinha\thanks{\texttt{makrand@berkeley.edu}. Supported by a Simons-Berkeley postdoctoral fellowship.}\\ \small Simons Institute and UC Berkeley \and 
 Ronald de Wolf\thanks{\texttt{rdewolf@cwi.nl}. Partially supported by the Dutch Research Council (NWO/OCW), as part of the Quantum Software Consortium programme (project number 024.003.037), and through QuantERA ERA-NET Cofund project QuantAlgo (680-91-034).}\\ \small{QuSoft, CWI and U.\ of Amsterdam}
}

\maketitle
\begin{abstract}
The Aaronson-Ambainis conjecture (Theory of Computing '14) says that every low-degree bounded polynomial on the Boolean hypercube has an influential variable. This conjecture, if true, would imply that the acceptance probability of every $d$-query quantum algorithm can be well-approximated almost everywhere (i.e., on almost all inputs) by a $\poly(d)$-query classical algorithm.
We prove a special case of the conjecture: in every completely bounded  degree-$d$ block-multilinear form with constant variance, there always exists a variable with influence at least $1/\poly(d)$. In a certain sense, such polynomials characterize the acceptance probability of quantum query algorithms, as shown by Arunachalam, Bri\"et and Palazuelos (SICOMP '19). As a corollary we obtain efficient classical almost-everywhere simulation for a particular class of quantum algorithms that includes for instance $k$-fold Forrelation.
Our main technical result relies on connections to free probability theory.
\end{abstract}

\thispagestyle{empty}
\clearpage
\newpage
\setcounter{page}{1}

\allowdisplaybreaks

\input{intro}

\input{prelims}

\input{proof}
\input{open}

\bibliographystyle{alpha}
\bibliography{refs}
\appendix
\input{free-probability}

\end{document}

%% file: intro.tex
\section{Introduction}

This paper is motivated by quantum query complexity and its relation to classical query complexity. Query complexity has been the context in which many of the main quantum algorithms have been developed, including Shor's \cite{Shor97} (building on \cite{Simon97}) and Grover's \cite{G96}. It has the added advantage that we actually know how to prove good lower bounds on query complexity, in contrast to a setting like circuit complexity. 

Quantum query complexity is closely connected to the study of bounded polynomials (or forms) on the Boolean hypercube. The key to this connection is that the amplitudes of the final state of a $d$-query quantum algorithm are polynomials of degree at most $d$ in the bits of the input~$x$, and therefore its acceptance probability $p(x)$ is a polynomial of degree at most $2d$. This observation was made by Beals, Buhrman, Cleve, Mosca and de Wolf~\cite{BBCMW01}, who used it to show that the bounded-error quantum query complexity and classical query complexity are polynomially related for any \emph{total} Boolean function. Since then a long line of research~\cite{ambainis:degreevsqueryj, ABK16, BHT17, ABBLSS17, T20, ABKRT20, BS21, SSW21} has tried to pinpoint the exact polynomial dependence as well as studied the relationship with other measures of complexity of a Boolean function (e.g., sensitivity, certificate complexity, and others~\cite{nisan&szegedy:degree,BdW02,BSS03}).

On the other hand, quantum algorithms can offer a huge (superexponential) advantage for partial functions, which are only defined on a subset of the Boolean hypercube; there are many known examples of partial functions whose classical query complexity is much larger than their quantum query complexity, for instance $k$-fold Forrelation and its variants \cite{AA:forrelation, T20, BS21, SSW21}. This means that the acceptance probability~$p(x)$ of a quantum algorithm cannot always be efficiently approximated  by a classical algorithm, since otherwise quantum algorithms could offer only a polynomial speedup for any function, be it total or partial.

However, we can set our sights lower, and ask whether it is possible to classically efficiently approximate $p(x)$ on \emph{almost all} inputs. 
The following conjecture, which first appeared in \cite{AA14} and is  attributed there to folklore, says that we can.

\begin{conjecture}[Folklore]\label{conj:folklore}
     The acceptance probability of any $d$-query quantum algorithm on $n$-bit inputs can be estimated up to additive error $\eps$ on a $1-\delta$ fraction of the inputs by a classical query algorithm making $\poly(d, 1/\epsilon, 1/\delta)$ queries. 
\end{conjecture}

This conjecture is one expression of the general idea that quantum computers can only give significant speedup (in terms of queries, circuit complexity, or other things) on very structured problems, i.e., when the input to the problem has a particular structure, for instance some periodicity or specific correlations between different parts of the input. For generic unstructured inputs, the conjecture says that only a limited quantum speedup can be expected.
This conjecture motivates and is implied by the following conjecture due to Aaronson and Ambainis~\cite{AA14}:

\begin{conjecture}[Aaronson-Ambainis conjecture]\label{conj:aa-inf}
     Let $f: \pmone^n \to [0,1]$ be a degree-$d$ multilinear polynomial. Then, the maximum influence among all variables in $f$ is at least $\poly(\Var[f], 1/d)$.
\end{conjecture}

The above conjecture poses a fundamental structural question about bounded polynomials on the hypercube and is a notable open problem in the analysis of Boolean functions. \conjref{conj:aa-inf} is known to hold if the function is Boolean-valued (this follows from \cite{M05,OSS05}). For bounded polynomials, \cite{AA14} observed that the results of Dinur, Friedgut, Kindler and O'Donnell~\cite{DFKO06} imply that the conjecture holds with at least an exponential dependence in $d$. Montanaro \cite{M12} proved a special case of the conjecture for \emph{block-multilinear forms} where all coefficients have the same magnitude\footnote{This argument can be generalized to the case when $n^{\Omega(d)}$ coefficients have the same magnitude and the rest are zero, as noted in \cite{M12} where the observation is  attributed to Ambainis.}. Defant, Masty\l{}o and Perez \cite{DMP18} generalized this to bounded polynomials where all Fourier coefficients have the same magnitude and showed that the conjecture holds with an $\exp(\sqrt{d\log d})$ dependence. O'Donnell and Zhao \cite{OZ16} showed that it is sufficient to prove the conjecture for so-called \emph{one-block decoupled} polynomials.

In this work, our motivation is to study \conjref{conj:aa-inf} for polynomials that represent the acceptance probability of quantum algorithms. Such polynomials have a lot more structure --- as shown by  Arunachalam, Bri\"et and Palazuelos~\cite{ABP18}, they can be represented in terms of \emph{completely bounded block-multilinear forms} (as described in the next section) and conversely, such forms even characterize quantum algorithms in a certain sense (see \secref{sec:open}). As such here we focus on understanding influences in such polynomials.

\subsection{Our results}

A degree-$d$ block-multilinear form $f(\x_1, \ldots, \x_d)$ mapping $\pmone^{n \times d}$ to $\BR$ is a polynomial where the variables are partitioned into $d$ blocks of $n$ variables each, and each monomial contains at most one variable from each block. Formally, $\x_b = (x_b(1), \ldots, x_b(n)) \in \pmone^n$ constitute the $b^\text{th}$ block of variables and 
\begin{equation}\label{eqn:tensor}
    f(\x_1,\ldots, \x_d) =  \E f +  \sum_{m = 1}^d \sum_{\substack{|\bb|=m\\|\bi|=m}} \hf_{\bb, \bi} \cdot  ~x_{b_1}(i_1)x_{b_2}(i_2)\cdots x_{b_m}({i_m}),
\end{equation}
where the tuple $\bb  = (b_1, \ldots, b_m)$ satisfies $1 \le b_1 < \ldots < b_m \le d$ and $\bi \in [n]^m$ is an $m$-tuple. Note that $m$
 is determined from the size of the tuple $\bb$, so we just write $\hf_{\bb, \bi}$ above.  

 Since each non-constant monomial contains at most one variable from each block and the ordering of the blocks is fixed, a degree-$d$ block-multilinear form $f: \pmone^{n \times d} \to \BR$ can be naturally viewed as a non-commutative polynomial in matrix variables with the constant term replaced with $\E f$ times the identity. Denoting the non-commutative polynomial as $f(\U_1, \ldots, \U_d)$ where each $\U_b = (U_b(1), \ldots, U_b(n))$ is a block of non-commutative variables, the completely bounded norm\footnote{The completely bounded norm originates in the theory of operator algebras. In the literature, this norm is sometimes defined for homogeneous block-multilinear forms only, but here we extend the definition to non-homogeneous block-multilinear forms.}  $\|f\|_{\cb}$ of the form $f$ is defined as
\[ 
\|f\|_{\cb} = \sup\Bigg\{\|f(\U_1, \ldots, \U_d)\|_{\op} ~\Bigg|~  N \in \N, U_b(i) \in \C^{N \times N}, \|U_b(i)\|_{\op} \le 1, b \in[d], i\in [n]\Bigg\}.
\]

 The supremum above is always attained and can be computed by solving a semidefinite program as shown by Gribling and Laurent~\cite{GL19}. One can also equivalently restrict the supremum in the definition above to unitary matrices since the convex hull of the set of unitary matrices is the unit operator norm ball. Moreover, $\|T\|_{\infty} \le \|T\|_{\cb}$ where $\|T\|_{\infty} = \max_{x \in \pmone^{n \times d}} |T(x)|$, so forms that are completely bounded  are also bounded on the hypercube. 

Our main result is a proof of the Aaronson-Ambainis conjecture for block-multilinear forms that are completely bounded. To state our result, we recall that the influence of a variable $x_b(i)$ on $f$ is
\[ \Inf_{b,i}(f) = \BE\left| \partial_{b,i} f(X)\right|^2, \]
where $X$ is uniform in $\pmone^{n \times d}$  and $\partial_{b,i}f(x)$ is the discrete derivative (see \secref{sec:prelims}). Denoting by $\maxinf(f) = \max_{b \in [d],i \in [n]} \inf_{b,i}(f)$ the maximum influence of any variable in $f$ and by $\Var[f]$ the variance of $f$ on the hypercube, we show: \\

\begin{theorem}\label{thm:aa}
     Let $f$ be a degree-$d$ block-multilinear form with $\|f\|_{\cb} \le 1$. Then, we have \[\maxinf(f)  \ge \dfrac{(\Var[f])^2}{e(d+1)^4}.\]
\end{theorem} 
    
The main technical ingredient in the proof of \thmref{thm:aa} is a new influence inequality for \emph{homogeneous} block-multilinear forms that relates the completely bounded norm to the influences.

\begin{restatable}{theorem}{bh}  {\rm (Non-commutative root-influence inequality).}
\label{thm:bh-intro}
Let $f$ be a homogeneous degree-$d$ block-multilinear form. Then, for blocks $b \in \{1,d\}$,

    \[ \cbnorm{f}  \ge \frac{1}{\sqrt{e(d+1)}} \sum_{i=1}^n \sqrt{\Inf_{b,i}(f)}. \]
\end{restatable}

\noindent{\bf Remark.} In general, the completely bounded norm can change if we permute the blocks, and the theorem above only gives a bound in terms of the influences of the variables in the leftmost and rightmost blocks.\\

The  inequality also easily implies the special case of \thmref{thm:aa} for homogeneous forms, with a better dependence on~$d$,  as follows:
    \begin{align*}
         \ \cbnorm{f}  &\ge \frac{1}{\sqrt{e(d+1)}} \sum_{i=1}^n \sqrt{\Inf_{b,i}(f)} \\
         \ &\ge \frac{1}{\sqrt{e(d+1)}}  \sum_{i=1}^n \frac{\Inf_{b,i}(f)}{\sqrt{\maxinf(f)}} \ge  \frac{\Var[f]}{\sqrt{e(d+1)\cdot  \maxinf(f)}},
    \end{align*}
    where the last inequality follows  since for any homogeneous block-multilinear form the sum of influences of variables in \emph{any} one block equals $\Var[f]$ (see \eqref{eqn:inf-tensor} in the preliminaries). Then, if $\cbnorm{f} \le 1$, it follows that \[
    \maxinf(f) \ge \frac{(\Var[f])^2}{e(d+1)}.
    \]
    The non-homogeneous case (\thmref{thm:aa}) requires a bit more care and we use the inequality as an intermediate step to prove \thmref{thm:aa} with a worse polynomial dependence on~$d$.

Combined with the results of \cite{AA14}, we obtain that completely bounded forms can be well-approximated by classical query algorithms (decision trees) on most inputs.

\begin{restatable}{corollary}{simulation}
\label{cor:sim}
 Let $\epsilon, \delta > 0$ and let $f: \pmone^{n \times d} \to \BR$ be a degree-$d$ block-multilinear form with $\|f\|_{\cb} \le 1$. Then, there is a deterministic classical algorithm that makes $O(d^5\epsilon^{-8}\delta^{-5})$ queries and approximates $f(x)$ up to an additive error $\epsilon$ on $1-\delta$ fraction of the inputs $x \in \pmone^{n \times d}$.
\end{restatable}

\subsubsection{Application to quantum algorithms}
\label{sec:quantum}

We consider quantum query algorithms of the type shown in \figref{fig:quantum}.  Any such algorithm has black-box access to the inputs $\x_1, \ldots, \x_d$ where $\x_b \in \pmone^{n}$ for each $b \in [d]$, via a phase oracle.  In other words, the algorithm can apply the unitary $O_{\x_b} = \Diag(\x_b)$ for each $b \in [d]$.  Note that $n$ here is the dimension of the underlying Hilbert space, and the inputs can be represented with $\log n$ qubits.

\begin{figure}[h]
    \centering
    \includegraphics[width=0.9\textwidth]{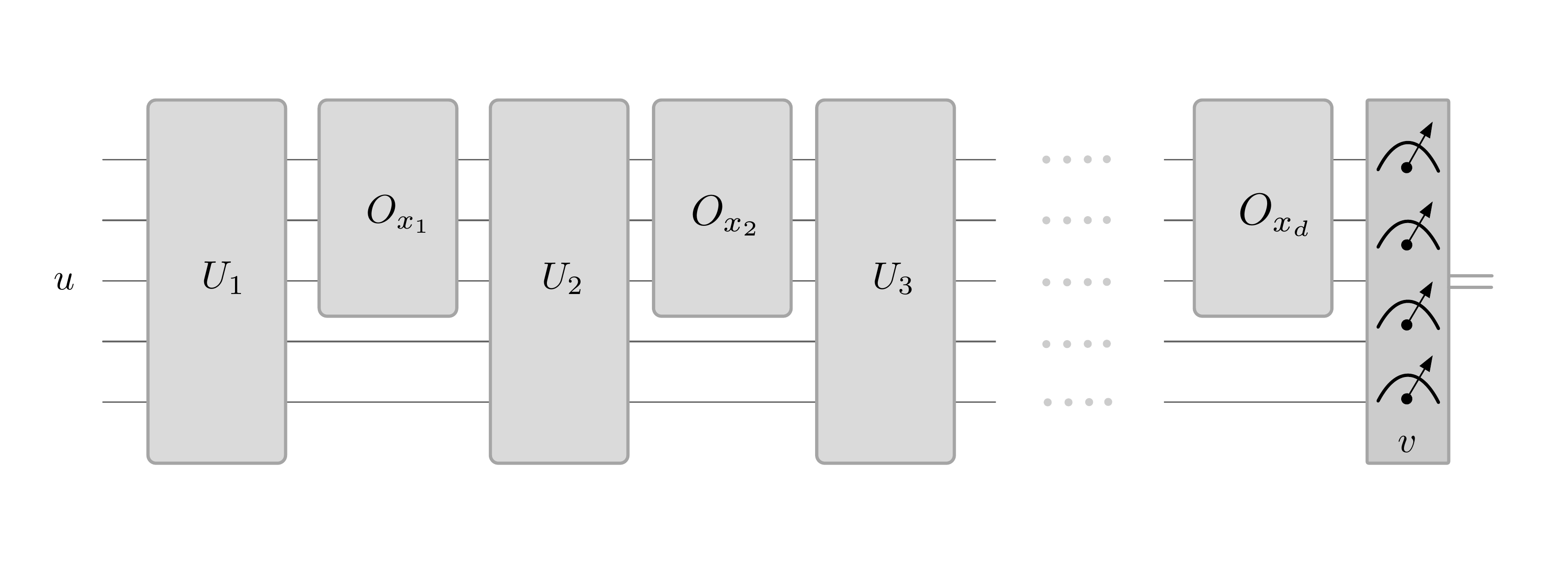}
    \caption{\footnotesize Quantum algorithms considered in \secref{sec:quantum}}
    \label{fig:quantum}
\end{figure}

The algorithm starts in some arbitrary quantum state\footnote{Throughout this paper, we will assume that all unitaries and states used in the quantum algorithm are real, which one may assume without loss of generality (see e.g.\ \cite{adh:qcomputability}).} $u \in \R^{n}$, makes $d$ quantum (phase) queries to oracles $O_{\x_b}$ for each $b \in [d]$, and succeeds according to a projective measurement that measures the projection of the final state onto some fixed state $v \in \R^n$. The algorithm is restricted to use each oracle $O_{\x_b}$ at most once. The inner product of the state $v$ with the final state at the end of the algorithm is given by the following degree-$d$ block-multilinear form $T : \pmone^{n \times d} \to \R$, 
\begin{equation}\label{eqn:q-tensor}
    T(\x_1,\ldots, \x_d) = u U_1(O_{\x_1} \otimes I_s)U_2 (O_{\x_2} \otimes I_s)U_3 \cdots (O_{\x_d} \otimes I_s) v,
\end{equation}
and the acceptance probability of the algorithm on input $x = (\x_1, \ldots, \x_d)$ is $T(x)^2$. 

The connection between such algorithms and completely bounded norm comes from the following  proposition in \cite{ABP18}.

\begin{proposition}[\cite{ABP18}, Theorem 3.2]
    Let $T : \pmone^{n \times d} \to \R$ be a degree-$d$ block-multilinear form given by \eqref{eqn:q-tensor}. Then, $\|T\|_{\cb} \le 1$. 
\end{proposition}

Using this connection, applying \corref{cor:sim} to $T$ implies the following almost-everywhere simulation result for quantum algorithms of the type mentioned above.

\begin{corollary}
     The acceptance probability of any $d$-query quantum algorithm of the type shown in \figref{fig:quantum} can be estimated up to an additive error $\eps$ on $1-\delta$ fraction of the inputs in $\pmone^{n \times d}$ by a classical query algorithm making $O(d^5\epsilon^{-8}\delta^{-5})$ queries. 
\end{corollary}

Note that quantum algorithms of the type considered in the above theorem can already exhibit super-exponential separation, in terms of the input size (which is $\log n$ qubits),  over classical algorithms in the query complexity model. For instance, problems like $k$-fold Forrelation (for $k= O(1)$) or its variants exhibit a $O(1)$ vs $n^{1-1/k}$ separation \cite{BS21, SSW21} between the quantum and classical query complexities.

\subsection{Proof overview}
\label{sec:bh}
\newcommand{\bp}{\mathbf{p}}
\newcommand{\bits}{\{0,1\}}
\newcommand{\addr}[1]{\mathrm{addr}(#1)}

We first consider the case of homogeneous forms and explain the key ideas that go towards proving \thmref{thm:bh-intro}. We can write a homogeneous block-multilinear form in the following way,
 \begin{align*}
    f(\x_1,\ldots, \x_d) &= \sum_{i_1, \ldots, i_d \in [n]} \hf_{i_1, \ldots, i_d} ~x_1(i_1)x_2(i_2)\cdots x_d({i_d}) \\
    \ & = \sum_{i=1}^n  x_1(i) \underbrace{\left(\sum_{i_2, \ldots, i_d \in [n]} \hf_{i_1, \ldots, i_d} ~x_2(i_2)\cdots x_d({i_d})\right)}_{\textstyle := f_i(\x_2,\ldots, \x_d)}.
\end{align*}
    For a first attempt, let us try to show that $\cbnorm{f}$ must be large by picking $x$ from the discrete cube $\pmone^{n \times d}$ as follows: for each block $b$ except the first block, we choose $\x_b$ uniformly and independently from $\pmone^n$, and for the first block we take $x_1(i) = \mathrm{sign}(f_i(\x_2,\ldots, \x_d))$. Taking expectation, this gives us that
 \[ 
 \cbnorm{f} \ge \sum_{i=1}^n \E |f_i| \ge {2^{-d/2}} \sum_{i=1}^n \|f_i\|_2,
 \]
 where the second inequality follows from the multilinear Khintchine inequality\footnote{The multilinear Khintchine inequality states that $\BE|f| \ge 2^{-d/2} \|f\|_2$ for a homogeneous degree-$d$ block-multilinear form $f$. A similar conclusion $\BE|f| \ge 3^{-d} \|f\|_2$ holds for any degree-$d$ polynomial $f$ on the hypercube and can be derived from the $(4,2)$-hypercontractive inequality $(\BE f^4)^{1/4} \leq  3^{d/2}\|f\|_2$ and using that $\E[|S|] \geq {\E[S^2]^{3/2}}/{\E[S^4]^{1/2}}$ for any random variable $S$.} which gives us an exponential dependence in $d$. Note that $\|f_i\|_2^2 = \Inf_{1,i}(f)$ for each $i$, thus we get that 
  \begin{equation}\label{eq:cbnormvsinf} \cbnorm{f} \ge {2^{-d/2}} \sum_{i=1}^n  \sqrt{\Inf_{1,i}(f)}. 
  \end{equation}
The above also gives a lower bound on $\|f\|_{\infty}$ which is also a lower bound on $\cbnorm{f}$. However, the  exponential dependence in $d$ is necessary for the sup-norm as the following example shows.

\begin{description}\item[Example.] Consider the following block-multilinear form closely related to the address function. Let $n=2^d$ and for $\ba = (a_1, \ldots, a_d) \in \bits^d$, let $\addr{\ba}$ denote the unique integer in $[n]$ whose binary expansion equals $\ba$. Define the degree-$(d+1)$ homogeneous block-multilinear form $f: \pmone^{n \times (d+1)} \to \{\pm1\}$ as follows,
\begin{align}\label{eqn:addr}
    \ f(\x_1, \ldots, \x_d, \x_{d+1}) &= \sum_{\ba \in \bits^d} g_{\ba}(\x_1,\ldots,\x_d) \cdot x_{d+1}(\addr{\ba}),
\end{align}
where $g_{\ba}(\x_1,\ldots, \x_d): \pmone^{n \times d} \to \{-1,0,1\}$ is defined as 
\[ g_{\ba}(\x_1,\ldots, \x_d) = \left(\frac{x_1(1) + (-1)^{a_1} x_1(2)}{2}\right)\cdot \cdots \left(\frac{x_d(1) + (-1)^{a_d} x_d(2)}{2}\right).\]

Note that $f$ only depends on the first two variables in the blocks $\x_1,\ldots, \x_d$ (which we refer to as the address blocks) and all the variables in the last block $\x_{d+1}$ (which we refer to as the data block). Moreover, $g_\ba(\x_1,\ldots,\x_d) \in \pmone$ iff the parity of bits in the address blocks matches with $\ba$, that is $x_b(1)x_b(2) = (-1)^{a_b}$ for every $b \in [d]$, and $g_\ba(\x_1,\ldots,\x_d) = 0$ otherwise. 

 It follows that $\|f\|_{\infty} = 1$, as for any setting of $\x_1, \ldots, \x_d$ exactly one term in the summation in \eqref{eqn:addr} survives. However, for $i = \addr{\ba} \in [n]$, 
\[ \inf_{d+1,i}(f) = \BE[|g_a|^2] = 2^{-d} = \frac1n,\] 
thus  $\sum_{i=1}^n \sqrt{\inf_{d+1,i}(f)} = \sqrt{n} = 2^{d/2}$. 
\end{description}

On the other hand, $\|f\|_{\cb} \ge 2^{d/2}$ in the example above --- this can be checked by plugging in the following values on the complex unit circle (one-dimensional unitaries): $x_b(1)=1,x_b(2)=i$ for each $b \in [d]$ and choosing the data block $\x_{d+1}$ so that all the magnitudes add up. Thus, one can hope that the freedom to choose large matrices can still allow us to show something like inequality~\eqref{eq:cbnormvsinf} for the completely bounded norm with a polynomial dependence on~$d$, instead of exponential.

\paragraph{Lower bounding $\cbnorm{f}$ using Haar random unitaries.} 
\begin{sloppypar}
Our key observation is that a non-commutative analog of the above strategy works very well. In particular, substituting $N \times N$ Haar random unitaries $U_2(1), \ldots, U_2(n), \ldots, U_d(1),\ldots, U_d(n)$ for the blocks $x_2, \ldots, x_d$ and choosing the block $x_1$ depending on the polar decomposition of $f_i(\U_2, \ldots, \U_d)$ allows one to obtain a much larger lower bound on the completely bounded norm $\cbnorm{f}$, losing only a polynomial rather than an exponential factor in~$d$. 
\end{sloppypar}

To obtain quantitative bounds, we need to understand the operator norm of low-degree polynomials of Haar random unitaries. A standard way to upper bound the expected operator norm of random matrices is via the trace method: computing the expected (normalized) trace of the matrix $(AA^*)^m$ for large enough $m$, and then taking $m$th root, gives a good control of the operator norm $\|A\|_{\op}$. Since the entries of a Haar random unitary are not independent of one another, it is hard to get a handle on the expected trace directly. A powerful method to understand such quantities is via free probability theory, which considers what happens when the dimension of the matrices $N \to \infty$. In this case, large random matrices behave like free operators, which live on an infinite-dimensional space with a corresponding ``trace''. We rely on a limiting theorem of Collins and Male \cite{CM11} who, by strengthening a result of Voiculescu \cite{V98}, show that the operator norm of a polynomial of Haar random unitaries converges to the operator norm of the polynomial of certain infinite-dimensional operators, called \emph{free Haar unitaries}; thus it suffices to study such free operators.

In free probability theory, such quantities have been studied for a long time (since the work of Haagerup \cite{H78}), and we rely on a result of Kemp and Speicher \cite{KS05} who generalized Haagerup's inequality  and showed that for free Haar unitaries one can obtain much better bounds for the operator norm using the usual trace method. In particular, one gets that almost surely as $N \to \infty$, we have
\[
\|f_i(\U_2,\ldots, \U_d)\|_{\op} \le \poly(d) \|f_i\|_2,
\]
in the non-commutative setting. Crucially, the improvement comes because free operators are much more constrained, and many terms that arise while looking at higher moments using the trace method in free probability are zero. One can keep close track of the non-zero terms by using careful combinatorial counting involving what are called \emph{non-crossing partitions}.  

Using the above, one can obtain \thmref{thm:bh-intro} with the strategy described above using the polar decomposition. The non-homogeneous case requires a bit more technical care, but the key underlying idea is the same.

%% file: prelims.tex
\section{Preliminaries}\label{sec:prelims}

\noindent{\bf Notation.}

 Throughout this paper, $[d]$ denotes the set $\{1,2,\dotsc, d\}$. For a random vector (or bit-string) $z$ in $\BR^n$, we will use $z_i$ or $z(i)$ to denote the $i$-th coordinate of $z$, depending on whether we need to use the subscript for another index. We shall use $\ui = (i_1,\ldots, i_d)$ for a $d$-tuple of indices.  For a $d$-tuple $\ui$, we write $|\ui|=d$ to denote the size of the tuple.

 For a matrix $M \in \C^{N \times N}$, we denote by $M^*$ its conjugate transpose.
 Given a string $x \in \R^N$, the $N \times N$ diagonal matrix with $x$ on the diagonal is denoted by $\Diag(x)$. The \emph{normalized trace} of an $N \times N$ matrix $M$ is defined as $\tr_N(M) = \frac{1}{N} \left(\sum_{i=1}^N M_{ii}\right)$. The operator norm of a matrix $M$ is denoted by $\|M\|_{\op}$. The left (resp.\ right) polar decomposition $(V, P)$ of a square matrix $M$ is a factorization of the form $M=VP$ (resp.\ $M=PV$) where $V$ is a unitary matrix and $P$ is a positive semidefinite matrix --- such a factorization always exists for any square matrix (it can be obtained easily from the singular-value decomposition of $M$). An $N \times N$ matrix $U$ is called a Haar random unitary if it is distributed according to the Haar measure on the Unitary group $\mathbb{U}(N)$.

Random variables are typically denoted by capital letters (e.g., $X$). We write $\BE[f(X)]$ and $\Var[f(X)]$ to denote the expectation and variance of the random variable $f(X)$ and if $f: \pmone^m \to \R$, we abbreviate it to $\E f$ and $\Var[f]$, where the expectation and variance are taken with respect to the uniform measure on the discrete cube $\pmone^m$.

\paragraph{Fourier Analysis on the Discrete Cube.}
\label{sec:fourier}
We give some basic facts about Fourier analysis on the discrete cube and refer to the book \cite{OD14} for more details. Every function $f: \{\pm1\}^m \to \BR$ can be written uniquely as a sum of monomials $\chi_S(x) = \prod_{i \in S} x_i$,
\begin{align}\label{eqn:expansion}
    f(x) = \sum_{S \subseteq [m]} \fhat(S) \chi_S(x),
\end{align}
where $\fhat(S) = \BE[f(X)\chi_S(X)]$ is the Fourier coefficient with respect to the uniform $X \in \{\pm 1\}^m$. The monomials $\chi_S(x) = \prod_{i \in S} x_i$ form an orthonormal basis for real-valued functions on $\{\pm1\}^m$, called the \emph{Fourier basis}. Parseval's identity implies that for uniform $X \in \{\pm 1\}^m$,
\begin{align*}
    \ \BE f^2 = \sum_{S \subseteq [m]} \fhat(S)^2, \text{ and } \Var[f] = \sum_{S \subseteq [m]: S \neq \emptyset} \fhat(S)^2.
\end{align*}

For a function on the hypercube, we define $\|f\|^2_2 := \BE f^2$ which can also be viewed as the sum of squared Fourier coefficients because of Parseval's identity.

The discrete derivative of a function on the hypercube $\{\pm 1\}^m$ is given by  
\[ \partial_i f(x) = \frac{1}{2} (f(x^{i\to 1}) - f(x^{i\to -1})),\]
where $x^{i\to b}$ is the same as $x$ except that the $i$-th coordinate is set to $b$. It is easily checked that $\partial_i f(x)$ coincides with the real partial derivative $\dfrac{\del{}}{\del x_i}$ of the real multilinear polynomial given by \eqref{eqn:expansion}. 

For a real-valued function $f: \pmone^m \to \BR$, the influence of a variable $x_i$ on $f$ is defined as 
\[ \inf_i(f) = \BE| \partial_if|^2 = \sum_{S \subseteq [m] : i \in S} \fhat(S)^2. \]

\paragraph{Block-multilinear Forms.}

A degree-$d$ block-multilinear form $f: \pmone^{n \times d} \to \BR$ is given by
\begin{equation}\label{eqn:tensor1}
    f(\x_1,\ldots, \x_d) =  \Ef + \sum_{m = 1}^d \sum_{\substack{|\bb|=m\\|\bi|=m}} \hf_{\bb, \bi} \cdot  ~x_{b_1}(i_1)x_{b_2}(i_2)\cdots x_{b_m}({i_m}),
\end{equation}
where the tuple $\bb  = (b_1, \ldots, b_m)$ satisfies $1 \le b_1 < \ldots < b_m \le d$ and $\bi \in [n]^m$ is an $m$-tuple. The expectation that yields the constant term is uniform over $\pmone^{n \times d}$. Note that $m$
 is determined from the size of the tuple $\bb$, so we just write $\hf_{\bb, \bi}$ above. 

From Parseval's identity, the variance of $f$ and the influence of $x_b(i)$ on $f$ (where $b \in [d]$ and $i \in [n]$) are respectively given by 
\[ \Var[f] =  \sum_{m=1}^d \sum_{\substack{|\bb|=m\\ |\bi|=m}} \hf_{\bb,\ui}^2 ~\text{ and }~ \inf_{b,i}(f) = \sum_{m=1}^d \sum_{\substack{|\bb|=m : b\in \bb \\|\bi|=m:i \in \bi}} \hf_{\bb,\ui}^2.\]

From the above, it follows that for any block $b \in [d]$, 
\begin{align}\label{eqn:inf-tensor}
  \  \sum_{i \in [n]} \inf_{b,i}(f)^2 \le \Var[f] \le \sum_{b \in [d], i \in [n]} \inf_{b,i}(f)^2 
\end{align}
where the first inequality is an equality if $f$ is a homogeneous degree-$d$ block-multilinear form. For any $b \in [d]$, we write $\maxinf_b(T) = \max \{ \inf_{b,i}(T)  \mid i \in [n]\}$ to denote the maximum influence of any variable in the block $\x_b$.

Note that if $f$ is a degree-$d$ block-multilinear form and if we fix some of the input bits to $\pm 1$, then the resulting function $g$ is also a degree-$d$ block-multilinear form with the same blocks $\x_1, \ldots, \x_d$, but it does not depend on the variables that were fixed. It is also easy to see that $\cbnorm{g} \le \cbnorm{f}$ because while computing $\|g\|_{\cb}$ we may restrict the matrix variables to $\pm I$ if that particular variable was set to $\pm 1$. In other words, completely bounded norm does not increase under restrictions.

%% file: proof.tex
\section{Influence in Completely Bounded Block-multilinear Forms}
\label{sec:proof}
\newcommand{\blocks}{\mathrm{blocks}}
\newcommand{\lt}{\mathrm{left}}
\newcommand{\rt}{\mathrm{right}}

In this section we prove the non-commutative root-influence inequality (\thmref{thm:bh-intro}),  the special case of the Aaronson-Ambainis conjecture given in \thmref{thm:aa}, and also briefly mention how the simulation result in \corref{cor:sim} follows from \thmref{thm:aa} and the results in \cite{AA14}. We first need some preliminaries from free probability theory.

\subsection{Low-degree Polynomials of Haar Random Unitaries}

As discussed in the proof overview, we require bounds on the operator norm (as well as normalized trace) of low-degree polynomials of random unitaries and these follow from known results in free probability theory. Here we explain these connections and also prove some auxillary lemmas needed for the proof of \thmref{thm:bh-intro} and \thmref{thm:aa}.

Let $z_{\ui}$ denote the non-commutative monomial $z_{i_1} z_{i_2} \cdots z_{i_d}$ for a $d$-tuple $\ui  = (i_1, \ldots, i_d) \in [t]^d$ and let $p(z_1, \ldots, z_t)$ be a non-commutative polynomial in the variables $z_1, \ldots, z_t$. We are interested in computing the operator norm $\|\cdot\|_{\op}$ and the normalized trace  $\tr_N$ of the polynomial $p(z_1, \ldots, z_t)$ (or its higher moments) when substituting $N \times N$ Haar random unitaries for the variables $z_i$.

As explained previously, the theory of free probability gives us tools that allow us to compute  the above in the limit $N \to \infty$. In particular, Voiculescu \cite{V98} showed that the  (normalized) trace of polynomials in Haar random unitaries and their conjugates converges to the trace of the same polynomial evaluated on certain infinite-dimensional operators called \emph{Haar unitaries} that satisfy a non-commutative notion of independence called \emph{free independence}. This was strengthened by Collins and Male \cite{CM11} who showed that such convergence also holds for the operator norm. A short primer on free probability is given in \appref{sec:free}, but for now one can think of $\CA$ as a self-adjoint algebra of bounded linear operators on a Hilbert space and $\phi$ as a trace functional for such operators in the statement given below.

\begin{theorem}[\cite{V98, CM11}] \label{thm:voiculescu}
    Let $p(z_1, \ldots, z_{2t})$ be a non-commutative polynomial in $\BR\langle z_1, \ldots, z_{2t}\rangle$. If $U_1, \ldots, U_t$ are $N \times N$ Haar random unitaries, then almost surely,
    \begin{align*}
     \ \tr_N[p(U_1, \ldots, U_t, U^*_1, \ldots, U_t^*)] &~\xrightarrow[N \to \infty]{}~ \phi[p(u_1, \ldots, u_t, u^*_1, \ldots, u^*_t)],\\
    \  \|p(U_1, \ldots, U_t, U^*_1, \ldots, U_t^*)\|_{\op} &~\xrightarrow[N \to \infty]{}~ \| p(u_1, \ldots, u_t, u^*_1, \ldots, u^*_t)\|,
    \end{align*}
    where $u_1, \ldots, u_t$ are free Haar unitaries in a $C^*$-probability space $(\CA, \phi)$ and $\|\cdot\|$ is the norm for the underlying $C^*$-algebra.
\end{theorem}

Using the above result it suffices to consider free Haar unitaries in a $C^*$-probability space to compute the operator norm and trace of polynomials of random unitaries. For a non-commutative polynomial $p(z_1, \ldots, z_t) = \sum_{|\ui|\le d} c_{\ui}z_{\ui}$, denoting by $\|p\|_2 =  \left(\sum_{|\ui| \le d} |c_{\ui}|^2\right)^{1/2}$, one can show the following easily using techniques from free probability. 

\begin{lemma} \label{thm:trace}
    Let $p(z_1, \ldots, z_t) = \sum_{|\ui|\le d} c_{\ui}z_{\ui} $ be a non-commutative degree-$d$ polynomial in $\R\langle z_1, \ldots, z_t\rangle$ and $u_1, \ldots, u_t$ be free Haar unitaries in a $C^*$-probability space $(\CA, \phi)$. Then, 
     \[ \phi[p(u_1, \ldots, u_t) (p(u_1, \ldots, u_t))^*] =  \|p\|_2^2.\]
\end{lemma}

The above implies that $\tr_N[p(U_1, \ldots, U_t) (p(U_1, \ldots, U_t))^*]$ converges to $\|p\|_2^2$ almost surely as $N \to \infty$. We shall defer the proof of \lref{thm:trace} to \appref{sec:app}, but to aid our intuition we note here that since the  $U_i$'s are independent $N \times N$ Haar random unitaries, the expected value

\[ \BE\left[\tr_N[p(U_1, \ldots, U_t) (p(U_1, \ldots, U_t))^*\right] = \|p\|_2^2,\] 
{and from concentration of measure, it is natural to expect that it converges to the above value}.

Similarly, to compute the operator norm of $p(U_1, \ldots, U_t)$ for Haar random unitaries one can instead study the norm of the polynomial evaluated on free Haar unitaries. Such bounds are easier to prove using the trace method since free independence imposes strong restrictions on the non-commutative moments. For instance, if $U_1$ and $U_2$ are independent $N \times N$ Haar random matrices, then $\BE[\tr_N(U_1U_2U^*_1U_2^*)]$ is non-zero (albeit quite small), while the corresponding trace evaluated on free Haar unitaries $u_1$ and $u_2$ is zero, that is $\phi(u_1u_2u^*_1u_2^*) = 0$. Thus, computing the trace $\phi[p(u_1,\ldots, u_t, u^*_1, \ldots, u_t^*)]$ reduces to handling the combinatorics of the patterns of $u_i$'s and $u_i^*$'s. 

In particular, we will rely on the following result that follows from the work of Kemp and Speicher \cite{KS05}  who consider the operator norm of homogeneous polynomials evaluated on free $R$-diagonal operators, a class that includes free Haar unitaries as well. We also remark that a bound where the right-hand side below is worse by a multiplicative $O(d^{1/2})$ factor also follows from the work of Haagerup\footnote{We note that Haagerup considered the more general case of polynomials in both $u_i$'s and $u^*_i$'s.}\cite{H78} who proved it in another context, predating even the introduction of free probability theory.

\begin{theorem}[\cite{KS05}]
\label{thm:kemp-speicher}
    Let $p(z_1, \ldots, z_t) = \sum_{|\ui| = d} c_{\ui}z_{\ui} $ be a homogeneous non-commutative degree-$d$ polynomial in $\R\langle z_1, \ldots, z_t\rangle$ and $u_1, \ldots, u_t$ be free Haar unitaries in a $C^*$-probability space. Then, 
    \[ 
    \|p(u_1, \ldots, u_t)\| \le \sqrt{e(d+1)} \cdot \|p\|_2,
    \]
    where the left-hand side denotes the norm in the underlying $C^*$-algebra. 
\end{theorem}

For completeness, we  introduce the necessary free probability background and some combinatorial details in \appref{sec:app}, and we present the fairly short proof of \thmref{thm:kemp-speicher} (from \cite{KS05}) there in a self-contained way. We shall need to extend the above bound to non-homogeneous polynomials. Let $p(z_1, \ldots, z_t) = \sum_{|\ui| \le d} c_{\ui}z_{\ui}$ and  let $p_k(z_1, \ldots, z_t) = \sum_{|\ui| = k} c_{\ui}z_{\ui}$ denote the degree-$k$ homogeneous part of $p$. Writing $p_k = p_k(u_1, \ldots, u_t)$ for $0 \le k  \le d$ and $p = p(u_1, \ldots, u_t)$, it follows from the triangle inequality,  \thmref{thm:kemp-speicher}, and Cauchy-Schwarz, that
    \begin{align*}
        \ \|p\| &\le \sum_{k=0}^d \|p_k\| 
        \le 
        \sum_{k=0}^d\sqrt{e(k+1)}\|p_k\|_2
        \le
       \sqrt{e}\left(\sum_{k=0}^d (k+1)\right)^{1/2} \left(\sum_{k=0}^d  \|p_k\|^2_2\right)^{1/2} \leq \sqrt{e}(d+1)  \cdot\|p\|_2.
    \end{align*}
Thus, we essentially get the same bound as in the homogeneous case, at the expense of an additional $O(d^{1/2})$ factor.

Collecting all the above we have the following as a direct consequence:

\begin{theorem} \label{thm:op-norm}
    Let $p(z_1, \ldots, z_t) = \sum_{|\ui|\le d} c_{\ui}z_{\ui} $ be a non-commutative degree-$d$ polynomial in $\R\langle z_1, \ldots, z_t\rangle$ and $U_1, \ldots, U_t$ be independent $N \times N$ Haar random unitaries. Then, as $N \to \infty$, the following holds almost surely, 
    \[ \tr_N[p(U_1, \ldots, U_t) (p(U_1, \ldots, U_t))^*] =  \|p\|_2^2,\]
    and
    \[ \|p(U_1, \ldots, U_t)\|_{\op} \le \sqrt{e}(d+1)  \cdot \|p\|_2,\]
    Moreover, the factor $(d+1)$ in the operator norm bound can be improved to $\sqrt{d+1}$ if the polynomial is homogeneous.
\end{theorem}

Based on the above theorem, we prove the following key lemma which captures the polar decomposition strategy mentioned in the earlier proof overview (\secref{sec:bh}). This will serve as the key ingredient in the proof of \thmref{thm:aa} and \thmref{thm:bh-intro}. 

\begin{lemma}\label{lem:polar}
    Let $p$ be a non-commutative degree-$d$ polynomial in $\R\langle y_1, \ldots, y_m, z_1, \ldots, z_t\rangle$ given by
    \[ p(y_1, \ldots, y_m, z_1, \ldots, z_t) = \sum_{i=1}^m y_i q_i(z_1, \ldots, z_t) + q_0(z_1, \ldots, z_t).\]
    Then, for every $\delta > 0$, there exist an integer $N$ and $N \times N$ unitaries $V_1,\ldots, V_m, W_1, \ldots, W_t$ such that 
    \[ \|p(V_1, \ldots, V_m, W_1, \ldots, W_t)\|_{\op} \ge \frac{1}{\sqrt{e}(d+1)} \sum_{i=1}^m \|q_i\|_2 - \delta.\]
    Moreover, the factor in front can be improved to $(e(d+1))^{-1/2}$ if $p$ is homogeneous. 
\end{lemma}

\begin{proof}[Proof of \lref{lem:polar}]
     For an arbitrary integer $N$, let us pick independent $N \times N$ Haar random unitaries $W_1, \ldots, W_t$ which we substitute for the variables $z_1,\ldots,z_t$, respectively, and let $M_i = q_i(W_1, \ldots, W_t)$ be the corresponding random matrices. Then, for any tuple of matrices $V_1, \ldots, V_m$ that we could substitute for the variables $y_1, \ldots, y_m$, we have that 
    \[ 
    p(V_1, \ldots, V_m, W_1, \ldots, W_t) = \sum_{i=1}^m V_i M_i + M_0.
    \] 
     \thmref{thm:op-norm} and union bound imply that as $N \to \infty$, with probability $1$ all the following events simultaneously hold: 
    \begin{itemize}
        \item $\|M_i\|_{\op} \le \sqrt{e}(d+1) \cdot \|q_i\|_2$ for each $i$,
        \item $\tr_N(M^*_iM_i) = \|q_i\|_2^2$ for each $i$, where $\tr_N(M)$ is the normalized trace.
    \end{itemize}
   To show that the operator norm must be large, let us fix a sufficiently large $N$ and a choice of $N\times N$ unitaries $W_1, \ldots, W_t$ such that $M_i$ satisfies $\|M_i\|_{\op} \le \sqrt{e}(d+1) \cdot \|q_i\|_2 + \epsilon$ and $\tr_N(M^*_iM_i) \ge \|q_i\|_2^2 - \epsilon$ for each $0\le i\le m$, where $\epsilon$ can be made arbitrarily small by increasing $N$. For $0 \leq i \leq m$, let $M_i = U_i P_i$ be the left polar decomposition of $M_i$, where $U_i$ is a unitary matrix and $P_i$ is a positive semidefinite matrix.
   
   We select the tuple of unitary matrices $V_1, \ldots, V_m$ that we substitute for the variables $y_1, \ldots, y_m$ to be $V_i = U_0U^*_i$ for $i \in [m]$. With this we have that $\|p(V_1, \ldots, V_m, W_1, \ldots, W_t)\|_{\op}$ is at least
    \begin{align*}
         \Big\|M_0 + \sum_{i=1}^m V_iM_i\Big\|_{\op} & = \Big\|U_0 P_0 + \sum_{i=1}^m U_0 U_i^* U_iP_i \Big\|_{\op} \\
        \ & =  \Big\|U_0 P_0 + \sum_{i=1}^m U_0 P_i\Big\|_{\op}  = \Big\| P_0 + \sum_{i=1}^m  P_i\Big\|_{\op}\ge \tr_N\Big(P_0 + \sum_{i=1}^m P_i\Big) \ge \tr_N\Big(\sum_{i=1}^m P_i\Big),
    \end{align*}
    where the last equality follows since the operator norm is unitarily invariant and the last two inequalities follow from the positive semidefiniteness of the $P_i$'s.

    For every positive semidefinite matrix $P$, we have that $\tr_N(P) \ge {\tr_N(P^2)}/{\|P\|_{\op}}$. 
  
    Hence,
     \[ \|p(V_1, \ldots, V_m, W_1, \ldots, W_t)\|_{\op} \ge \sum_{i=1}^m \frac{\tr_N(P_i^2)}{\|P_i\|_{\op}}.\]
     By our choice of $M_i$, we have that $\tr_N(P_i^2) = \tr_N(M_i^* M_i) \ge \|q_i\|_2^2 - \eps$ and $\|P_i\|_{\op} = \|M_i\|_{\op} \le \sqrt{e}(d+1)\|q_i\|_2 + \eps$. Since $\eps$ can be made arbitrarily small by increasing $N$, it follows that 
      \[ \|p(V_1, \ldots, V_m, W_1, \ldots, W_t)\|_{\op} \ge \frac1{\sqrt{e}(d+1)} \sum_{i=1}^m \|q_i\|_2 - \delta ,\]
     for large enough $N$. The improved bound for the homogeneous case follows directly by plugging the bound of \thmref{thm:op-norm} into the above proof.
\end{proof}

\subsection{Non-commutative root-influence inequality}
\label{sec:bh-proof}

For clarity in the proofs below, we remind our  convention that all tuples or blocks are denoted with boldface fonts (e.g. $\BU_1$ or $\BA$), while a single element is denoted without boldface (e.g. $U_1(i)$ or $A_i$ or $A$). Before proceeding with the proof, we restate the statement for convenience.

\bh*

\begin{proof}[Proof of \thmref{thm:bh-intro}] 
Since $f$ is homogeneous, we can write
   \begin{align*}
    f(\x_1,\ldots, \x_d) &= \sum_{i_1, \ldots, i_d \in [n]} \hf_{i_1, \ldots, i_d} ~x_1(i_1)x_2(i_2)\cdots x_d({i_d}) \\
    \ & = \sum_{i=1}^n  x_1(i) \underbrace{\left(\sum_{i_2,\ldots, i_d \in [n]} \hf_{i_1, \ldots, i_d} ~x_2(i_2)\cdots x_d({i_d})\right)}_{\textstyle := f_i(\x_2,\ldots, \x_d)}.
\end{align*}
 In this case, it follows from \eqref{eqn:inf-tensor} that for each $i \in [n]$, we have 
 \begin{equation}\label{eqn:var}
     \ \Var[f_i] = \|f_i\|^2_2 = \inf_{1,i}(f) \text{ and }  \Var[f] = \sum_{i=1}^n \inf_{1,i}(f).
 \end{equation}

  Let us denote the corresponding non-commutative block-multilinear polynomials by $f(\BU_1, \ldots, \BU_d)$ and $f_i(\BU_2, \ldots,\BU_d)$ where $\BU_b = (U_b(1), \ldots, U_b(n))$ denotes the $b^\text{th}$ block of non-commutative variables. To show a lower bound on $\cbnorm{f}$ it suffices to exhibit a collection of square matrices $\{U_b(i)\}_{b\in [d], i \in [n]}$ with operator norm at most~1, such that $\|f(\BU_1, \ldots, \BU_d)\|_{\op}$ is large.

Applying \lref{lem:polar} for the homogeneous case (with $p = f$, $q_i=f_i$ for $i \in [n]$, and $q_0=0)$, it follows that for every $\delta > 0$ there exists an integer $N$ and a choice of tuples of $N \times N$ unitaries $\BU_1, \ldots, \BU_d$ such that  
      \[ \cbnorm{f} \ge \|f(\BU_1, \ldots, \BU_d)\|_{\op} \ge \frac1{\sqrt{e(d+1)}} \sum_{i\in [n]} \|f_i\|_2  -\delta \stackrel{\eqref{eqn:var}}{\ge}  \frac{1}{\sqrt{e(d+1)}} \left(\sum_{i=1}^n \sqrt{\Inf_{1,i}(f)} \right) -\delta.\]
Taking $\delta \to 0$, we get the statement of the lemma. The proof for the inequality when $b=d$ is the last block follows similarly by using the right polar decomposition.
\end{proof}

\subsection{Aaronson-Ambainis Conjecture for non-homogeneous forms}

In this section, we prove \thmref{thm:aa}, which requires handling non-homogeneous forms. The proof will be similar to the proof of \thmref{thm:bh-intro} but we will need to be careful about certain details. 

\begin{proof}[Proof of \thmref{thm:aa}]
Any block-multilinear polynomial $f(x_1, \ldots, x_d)$ can be written as 
\begin{align*}
    f(\x_1,\ldots, \x_d) &= \BE f + \sum_{b\in [d]} f_b(\x_b, \x_{b+1}, \ldots, \x_d),
\end{align*}
where $f_b$ consists of all monomials of $f$ that start with a variable in the $b^\text{th}$ block $\x_b$. Note that $f_b$ depends only on the variables in blocks $\x_b, \x_{b+1},\ldots, \x_d$. Moreover, it follows from \eqref{eqn:inf-tensor} that 
 \begin{equation}\label{eqn:var-general}
     \ \Var[f] = \sum_{b \in [d]} \|f_b\|_2^2 = \sum_{b \in [d]} \Var[f_b],
 \end{equation}
so there exists a block $\beta \in [d]$ such that $\Var[f_{\beta}] \ge \frac{1}{d}\Var[f]$. 

Since $f_{\beta}$ contributes a lot to the variance, it is natural to try to find an influential variable in the block $\x_{\beta}$. Towards this end,  we pull out the variables $x_{\beta}(i)$ and write
\begin{align*}
    f_{\beta}(\x_{\beta},\ldots, \x_d) &= \sum_{i\in [n]} x_{\beta}(i) f_{\beta,i}(\x_{\beta+1}, \ldots, \x_d),
\end{align*}
for block-multilinear polynomials $f_{\beta,i}(\x_{\beta+1}, \ldots, \x_d)$. Note that some of the $f_{\beta,i}$'s could be identically zero, so let us define $S$ to be the set of those $i$ such that $f_{\beta,i}$ is non-zero. We note that
\begin{align} \label{eqn:part-inf}
  \|f_{\beta,i}\|_2^2  =  \Inf_{\beta,i}(f_{\beta}) \le \Inf_{\beta,i}(f)  
\end{align}
which implies that
\begin{align}\label{eqn:var-main}
    \frac{1}{d} \Var[f] \le \Var[f_{\beta}] = \sum_{i \in S}\|f_{\beta,i}\|_2^2 = \sum_{i \in S} \Inf_{\beta,i}(f_{\beta}).
\end{align}
\begin{sloppypar}
Denote the corresponding non-commutative block-multilinear polynomials by $f(\BU_1, \ldots, \BU_d)$,  $f_b(\BU_{b}, \ldots,\BU_d)$, and $f_{\beta}(\BU_{\beta+1}, \ldots,\BU_d)$ where $\BU_b = (U_b(1), \ldots, U_b(n))$ denotes the $b^\text{th}$ block of non-commutative variables. To show a lower bound on $\cbnorm{f}$ it suffices to exhibit a collection of square matrices $\{U_b(i)\}_{b\in [d], i \in [n]}$ with operator norm at most~1 such that $\|f(\BU_1, \ldots, \BU_d)\|_{\op}$ is large.
\end{sloppypar}
  
 We set the matrices in blocks $\BU_1, \ldots, \BU_{\beta-1}$ to be zero (that is, the all-zero matrix $\BZ$). Note that with this choice all polynomials $f_b(\U_b, \ldots, \U_d)$ where $b < \beta$ vanish and the non-commutative polynomial becomes 
 \[ f(\BZ, \ldots, \BZ, \BU_{\beta}, \BU_{\beta+1}, \ldots, \BU_d) = \sum_{i\in S} U_{\beta}(i) f_{\beta,i}(\BU_{\beta+1}, \ldots, \BU_d) + \sum_{b=\beta+1}^d f_b(\BU_b, \BU_{b+1}, \ldots, \BU_d) + \Ef,\]
  which is a non-commutative polynomial of the form considered in \lref{lem:polar} (with $m = |S|$, $q_i = f_{\beta,i}$ and $q_0 = \sum_{b=\beta+1}^d f_b + \Ef$). Thus, by \lref{lem:polar} for every small $\delta>0$ there exists an integer $N$ and a choice of $N \times N$ matrices for the blocks $\BU_{\beta},\ldots, \BU_d$ such that 
        \begin{align*}
             \ \cbnorm{f} & \ge \|f(\BZ, \ldots, \BZ, \BU_{\beta}, \BU_{\beta+1}, \ldots, \BU_d)\|_{\op} & \\
             \  & \ge \frac1{\sqrt{e}(d+1)} \sum_{i\in S} \|f_{\beta,i}\|_2 -\delta  \stackrel{\eqref{eqn:part-inf}}{=}  \frac{1}{\sqrt{e}(d+1)} \left(\sum_{i \in S} \sqrt{\Inf_{\beta,i}(f_{\beta})} \right) -\delta & \\
             \ &\stackrel{\eqref{eqn:var-main}}{\ge}  \frac{1}{\sqrt{e}(d+1)} \left( \frac{\sum_{i \in S} \Inf_{\beta,i}(f_{\beta})}{\sqrt{\maxinf(f)}} \right) -\delta  \stackrel{\eqref{eqn:part-inf}}{\ge}  \frac{1}{\sqrt{e}(d+1)^{2}} \left( \frac{\Var[f]}{ \sqrt{\maxinf(f)}} \right) -\delta
        \end{align*}
        Taking $\delta \to 0$ and using the assumption that $\|f\|_{\cb} \le 1$, we obtain the statement of the theorem:
     \[
     1\geq \cbnorm{f} \ge \frac{1}{\sqrt{e}(d+1)^{2}} \cdot \frac{\Var[f]}{\sqrt{\maxinf(f)}} \implies \maxinf(f) \ge  \frac{(\Var[f])^2}{e(d+1)^4}. \qedhere
     \]
\end{proof}

\subsection{Approximating completely bounded forms with decision trees}

In this section, we briefly mention how to obtain \corref{cor:sim}.
Aaronson and Ambainis \cite[Theorem 3.3]{AA14} showed that querying the most influential variable reduces the variance of the function~$f$, and if that influence is lower bounded by a polynomial in $\Var[f]/d$, then after $\poly(d)$ queries (the exact quantitative dependence can be read off from their proof), the variance of the function becomes small enough so that it can be approximated almost-everywhere by its expectation.  Since the family of degree-$d$ block-multilinear forms with completely bounded norm at most one is closed under restrictions, one can apply \thmref{thm:aa} repeatedly. This gives us \corref{cor:sim}.

%% file: open.tex
\section{Discussion and Open Problems} 
\label{sec:open}

To prove \conjref{conj:folklore} in full generality, one would need to consider arbitrary quantum query algorithms: such an algorithm operating on an input $z \in \pmone^m$ always makes queries to the same oracle $O_z$ (with a control qubit possibly). One can always convert any such algorithm to the type given in \figref{fig:quantum} by replacing the oracle $O_z$ used at each step $b$ with a new oracle $O_{\x_b}$ where $\x_b \in \pmone^{m+1}$. The execution of the original algorithm can then be recovered by substituting $\x_b = (z,1)$ for every $b \in [d]$.
As such one can always obtain a completely bounded block-multilinear form associated with any quantum query algorithm. Conversely, the work \cite{ABP18} shows that the existence of a degree-$2d$ homogeneous block-multilinear form $F : \pmone^{(m+1) \times 2d}$   with completely bounded norm at most one also implies the existence of a $d$-query quantum algorithm whose bias is given by $F((z,1), \ldots, (z,1))$ on every input $z \in \pmone^m$. Thus, completely bounded homogenous block-multilinear forms fully characterize quantum query algorithms in this sense.

In many works in quantum query complexity that concern worst-case complexity, understanding completely bounded or bounded block-multilinear polynomials is sufficient to prove lower bounds as well as give worst-case classical simulation results (i.e.\ for all inputs), see for instance \cite{AA:forrelation, BGGS21}. However, a transformation that converts a general quantum query algorithm to the type shown in \figref{fig:quantum} is not conducive to the almost-everywhere results considered in this paper, as the size of the input domain increases exponentially and the number of relevant inputs (i.e.\ where each $\x_b$ is set to the same $(z,1)$) becomes an exponentially small fraction of the new domain. 

It thus remains an intriguing open problem to see if the characterization of \cite{ABP18} can be used to make further progress on \conjref{conj:folklore}. One can also hope to make progress on \conjref{conj:folklore} without relying on the connection via influences ---  recently, Aaronson, Ingram and Kretschmer \cite{AIK21} managed to directly prove \conjref{conj:folklore} for the special case where the quantum algorithm queries a sparse oracle, without first proving a special case of \conjref{conj:aa-inf}.

Another interesting direction is to show that the Aaronson-Ambainis conjecture holds for bounded block-multilinear polynomials, that is, polynomials whose sup-norm on the Boolean hypercube is at most one. While this by itself does not suffice for the application to quantum algorithms as explained above, it might pave the way towards \conjref{conj:aa-inf} in full generality. Lastly, the free-probability toolbox has already found several applications in quantum information theory (see e.g.~\cite{Yin:freeprob,CollinsNechita}), and we hope this work will stimulate more applications elsewhere as well.

\paragraph{Acknowledgments.} We thank Scott Aaronson, Srinivasan Arunachalam, Jop Bri\"et and Ryan O'Donnell for helpful comments.

%% file: free-probability.tex
\section{Free Probability Primer}
\label{sec:app}

There are many excellent books on free probability theory. In particular, we refer to the book \cite{NS06} for more details than the brief introduction given here.

\subsection{Preliminaries}
\label{sec:free}

\paragraph{$C^*$-algebras.} Let $\CA$ be a unital $C^*$-algebra. For our purposes, we can think of this as an algebra of bounded operators on a complex Hilbert space which is self-adjoint ($a \in \CA$
 implies $a^* \in \CA$), closed in the operator norm $\|\cdot\|$, and contains the identity ($\bone \in \CA$). A faithful trace $\phi$ on $\CA$ is a continuous linear functional $\phi: \CA \to \C$ that is  unital ($\phi(\bone)=1$), positive $\phi(aa^*) \ge 0$, and $\phi(aa^*) = 0$ iff $a=0$. 
 
 The pair $(\CA, \phi)$ where $\CA$ is a unital $C^*$-algebra and $\phi$ is a faithful trace on $\CA$ is called a $C^*$-\emph{probability space}. Elements of $\CA$ are called non-commutative random variables. An example of a $C^*$-probability space is the class $(M_n(\C), \tr_n)$, which is the class of $n \times n$ complex matrices with the normalized trace  functional defined as $\tr_n(M) = \frac{1}{n} \sum_{i=1}^n {M_{ii}}$. General $C^*$-probability spaces allow us to extend these definitions to infinite-dimensional operators, which are needed to define a non-commutative analog of independence called \emph{free independence}. Faithfulness of the trace $\phi$ then ensures that $\|a\|= \lim_{m \to \infty} \phi((aa^*)^m)^{1/2m}$ (see \cite[Proposition 3.17]{NS06}). In particular, this allows one to compute the norm $\|\cdot\|$ by using the trace method and taking higher powers of the trace functional $\phi$, as we will see below.

\paragraph{Free Independence.} Let $(\CA, \phi)$ be a $C^*$-probability space and let $\{\CA_i\}_{i=1}^n$ be unital $*$-subalgebras of  $\CA$. They are said to be \emph{free} (or \emph{freely independent}) if for all $k \in [n]$, for all indices $i_1, \ldots, i_k \in [n]$, and for all $a_1 \in \CA_{i_1}, \ldots, a_k \in \CA_{i_k}$ satisfying $\phi(a_1)=\ldots =\phi(a_k)=0$, the joint \emph{free moment},
\[ \phi(a_1 \cdots a_k) = 0\]
whenever $j_1 \neq j_2, j_2\neq j_3, \ldots, j_{k-1} \neq j_k$, that is, the free moments vanish when all the neighboring elements in the sequence $a_1, \ldots, a_k$ come from  subalgebras with distinct indices, for example, $\phi(a_1a_2a^*_1a^*_2a_3a_2)=0$.

{Non-commutative random variables $a_1, \ldots, a_n \in (\CA, \phi)$ are said to be free if the subalgebras $\{\CA_i\}_{i=1}^n$ are free, where $\CA_i$ is the unital $*$-subalgebra  generated by $a_i$ (the linear span of all monomials $a^{\eps_1}_ia^{\eps_2}_i\cdots a^{\eps_r}_i$ where $\eps_1, \ldots, \eps_r \in \{1,*\}$ and $r \in \N \cup \{0\}$). Note that the corresponding unital $C^*$-subalgebras obtained by taking the norm closure of each $\CA_i$ are also freely independent in this case (see \cite[Exercise 5.23]{NS06}).}

We remark that the set of free non-commutative random variables is an empty set if the underlying $C^*$-probability space is finite (for instance $(M_n(\C), \tr_n)$), so to find non-trivial examples one needs to work with infinite-dimensional $C^*$-probability spaces. 

\paragraph{Free Haar Unitaries and Free Groups.} Let $(\CA, \phi)$ be a $C^*$-probability space. An element $u \in \CA$ is a \emph{Haar unitary} if it is a unitary, i.e.\ $uu^*=u^*u= \bone$, and if $\phi(u^k) = 0$ for all non-zero integers $k$. A family $S = \{u_1, \ldots, u_n\} \in \CA$ in a $C^*$-probability space $(\CA, \phi)$ is called a \emph{free Haar unitary family} if each $u \in S$ is a Haar unitary and if $u_1, \ldots, u_n$ are free. For notational convenience, let us define $S^* = \{u^*_1, \ldots, u^*_n\}$ to be the set of corresponding adjoints.

One can give a very precise condition when the trace $\phi$ evaluated on a non-commutative monomial in the $u_i$'s vanishes in terms of the free group. The \emph{free group} $F_n$ with generating set $S$ is an infinite discrete group constructed as follows: a word is defined to be product of elements of $S \cup S^*$ with $\bot$ denoting the empty word that contains no symbols. A word is called reduced if it does not contain a sub-word of the form $g g^{*}$ or $g^{*} g$ for $g \in S$. Given a word that is not reduced, the process of repeatedly removing such sub-words until it becomes reduced is called reduction. The free group $F_n$ consists of all reduced words that can be built from the symbols in $S \cup S^*$ with the group operation being a product of words followed by reduction. The identity is the empty word $\bot$. 

For a $d$-tuple $\bi  = (i_1, \ldots, i_d) \in [m]^d$, let $u_{\bi}$ denote the non-commutative monomial $u_{i_1}\cdots u_{i_d}$ and write $u^*_{\bi} = (u_{\bi})^* = u^*_{i_d} \cdots u^*_{i_1}$. Let $\bi_1,\ldots, \bi_t, \bj_1, \ldots, \bj_t$ each be a $d$-tuple in $[m]^d$ and consider the degree-$2td$ non-commutative monomial $w = u^{}_{\bi_1}u^*_{\bj_1} u^{}_{\bi_2}u^*_{\bj_2}\cdots u^{}_{\bi_t}u^*_{\bj_t}$. Note that a degree-$2td$ monomial $w$ corresponds to an ordered $2td$-tuple of variables. To illustrate, if $t=1, m=3$ and $\bi_1=(1,2,3)$ and $\bj_1 = (2,2,1)$, then $w = u_{1}u_2u_3(u_2u_2u_1)^* = u_{1}u_2u_3u_1^*u^*_2u^*_2$ and corresponds to the ordered tuple $(u_1, u_2, u_3, u^*_1, u^*_2, u_2^*)$. We can also interpret $w$ as a word in the free group by applying the reduction rules. Then the next proposition follows from the definitions of free independence and Haar unitaries.

\begin{proposition} \label{prop:word}
    $\phi(w) = 1$ iff $w$ reduces to identity in the free group $F_n$, and $\phi(w) = 0$ otherwise.
\end{proposition}

 For a monomial $w$ that reduces to identity in the free group, the procedure for reducing a monomial $w$ as above first removes some adjacent pair $u_k$ (at index $i$) and $u^*_k$ (at index $j)$, then removes another adjacent pair $u_l$ and $u^*_l$ in the resulting word and so on and so forth until we reach the empty word. In particular, this reduction procedure produces a pairing of the set $[2td]$ where the index $i$ and $j$ are paired up iff the  variables at indices $i$ and $j$ in the monomial $w$ are $u_k$ and $u^*_k$ (for some $k$). Moreover, this pairing is what is called a \emph{non-crossing} pairing defined below (see \figref{fig:non-crossing}). Note that a monomial could be reduced to identity in different ways, so there could be many such non-crossing pairings for a given monomial $w$.

\begin{figure}[h!]
    \centering
   \includegraphics[width=0.6\textwidth]{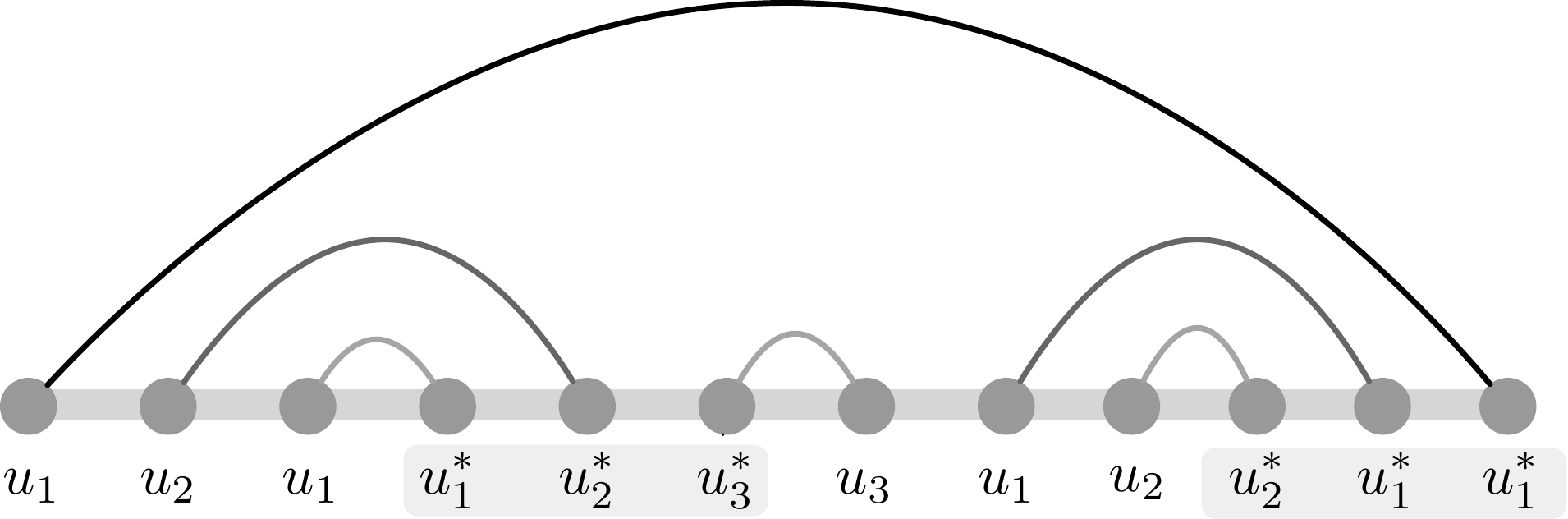}
   \caption{\footnotesize A non-crossing $*$-pairing resulting from the reduction of a word to identity in the free group}
    \label{fig:non-crossing}
\end{figure}

\paragraph{Non-crossing Pairings.}  For any even integer $n$, let $\CP_2(n)$ denote the set of all pairings of $n$, that is, the set of all partitions of $[n]$ where each block is of size two. Let $\NC_2(n) \subseteq \CP_2(n)$ denote the set of all  pairings of $[n]$ that are non-crossing,\emph{ i.e.} pairings  which do not contain blocks $\{i_1,i_3\}, \{i_2, i_4\}$ such that $i_1 < i_2 < i_3 < i_4$. 

For integers $d,m$, we divide the set $[2dm]$ into $2m$ consecutive blocks of $d$ elements each and color consecutive blocks alternatively with red and blue. Formally, for $i \in [2m]$, the elements $\{(i-1)d+1,\ldots, id\}$ are colored red if $i$ is odd and blue if $i$ is even. We define $\NC_2^*(d,m) \subseteq \NC_2(2dm)$ to be the set of those non-crossing pairings of $[2dm]$ which only pair up elements of different colors. We call any pairing in $\NC_2^*(d,m)$ a $*$-pairing.

We shall need the following combinatorial fact about the number of $*$-pairings (see \cite[Corollary 3.2]{KS05}).

\begin{lemma}\label{lem:catalan-fuss}
    For all $d, m$, the number of $*$-pairings  $|\NC_2^*(d,m)|$ equals the Fuss-Catalan number
    \[ 
    C_{d,m} = \frac{1}{m}\binom{m(d+1)}{m-1} = O\left(\frac{(d+1)^{m(d+1)}}{\left(d+\frac1m\right)^{md+1}}\right).
    \] 
\end{lemma}

\subsection{Proofs of \lref{thm:trace} and \thmref{thm:op-norm}}

\begin{proof}[Proof of \lref{thm:trace}]
    Writing $u^*_\bi = (u_\bi)^*$ for a tuple $\bi$ and using linearity of $\phi$, we have that 
    \[ 
    \phi[p(u_1, \ldots, u_t) (p(u_1, \ldots, u_t))^*] =  \sum_{|\bi|,|\bj| \le d} c_{\bi}c_{\bj} \phi(u_{\bi}u^*_{\bj}).
    \]
    From \pref{prop:word}, the term $\phi(u_{\bi}u^*_{\bj})$ is 1 iff $u_{\bi}u^*_{\bj}$ reduces to identity in the free group $F_t$ with generators $u_1, \ldots, u_t$. For the right-hand side above, this only happens when $\bi=\bj$ and thus these are the only non-zero terms. Thus, 
    \[ \phi[p(u_1, \ldots, u_t) (p(u_1, \ldots, u_t))^*] =  \sum_{|\bi| \le d} |c_{\bi}|^2. \qedhere\]
\end{proof}

Below we present the argument of Kemp and Speicher \cite{KS05}. Our exposition follows their proof closely but we adapt it to our context.
\begin{proof}[Proof of \thmref{thm:kemp-speicher}]
     We have that $\|p\| = \lim_{m \to \infty} \left(\phi((pp^*)^m)\right)^{1/(2m)}$ by the faithfulness of the trace $\phi$. Writing $u^*_\bj = (u_\bj)^*$ for a tuple $\bj$, we can compute 
    \begin{align*}
         \phi((pp^*)^m) = \sum_{\substack{|\bi_1|=\ldots=|\bi_m|=d\\|\bj_1|=\ldots=|\bj_m|=d}} c_{\bi_1}\cdots c_{\bi_m}c_{\bj_1}\cdots c_{\bj_m} \phi(u^{}_{\bi_1}u^*_{\bj_1} \cdots u^{}_{\bi_m}u^*_{\bj_m}).
    \end{align*}
    
    Since $u_1, \ldots, u_t$ are free Haar unitaries, \pref{prop:word} implies that $\phi(u^{}_{\bi_1}u^*_{\bj_1} \cdots u^{}_{\bi_m}u^*_{\bj_m})$ is 1 iff the word $u^{}_{\bi_1}u^*_{\bj_1}\cdots u^{}_{\bi_m}u^*_{\bj_m}$ reduces to identity in the free group $F_t$, and is 0 otherwise. Moreover, if the word corresponding to the index $(\bi_1,\bj_1, \ldots, \bi_m,\bj_m)$ reduces to identity, then there exists a $*$-pairing $\pi \in \NC_2^*(d,m)$ which matches only variables with the same indices. We call any such $*$-pairing $\pi$ consistent with the $2dm$-tuple $(\bi_1,\bj_1, \ldots, \bi_m,\bj_m)$ and denote this by the indicator function $\ind[\pi, \bi_1,\bj_1, \ldots, \bi_m,\bj_m]$.
    
    The above implies that we may bound 
    \[ \phi(u^{}_{\bi_1}u^*_{\bj_1} \cdots u^{}_{\bi_m}u^*_{\bj_m})  \le \sum_{\pi \in \NC^*_2(d,m)} \ind[\pi, \bi_1,\bj_1, \ldots, \bi_m,\bj_m],\]
    where the inequality occurs because there could be multiple $*$-pairings consistent with a tuple. We thus have that 
    \begin{align*}
         \phi((pp^*)^m) &\le \sum_{\substack{|\bi_1|=\ldots=|\bi_m|=d\\|\bj_1|=\ldots=|\bj_m|=d}} c_{\bi_1}\cdots c_{\bi_m}c_{\bj_1}\cdots c_{\bj_m} \sum_{\pi \in \NC^*_2(d,m)} \ind[\pi, \bi_1,\bj_1, \cdots, \bi_m,\bj_m]\\
         &= \sum_{\pi \in \NC^*_2(d,m)}  \sum_{\substack{|\bi_1|=\ldots=|\bi_m|=d\\|\bj_1|=\ldots=|\bj_m|=d}} c_{\bi_1}\cdots c_{\bi_m}c_{\bj_1}\cdots c_{\bj_m} \ind[\pi, \bi_1,\bj_1, \cdots, \bi_m,\bj_m].
    \end{align*}
    
    If a term corresponding to a fixed $*$-pairing $\pi$ is non-zero, then the list of indices $(\bi_1, \ldots, \bi_m)$ is the same as $(\bj_1, \ldots, \bj_m)$ up to the exact ordering. Let us relabel $(\bi_1, \ldots, \bi_m) = (a_1, \ldots, a_{dm})$ and $(\bj_1, \ldots, \bj_m) = (b_1, \ldots, b_{dm})$ and let $c_{a_1, \ldots, a_{dm}} = c_{\bi_1}\ldots c_{\bi_m}$ and $c_{b_1, \ldots, b_{dm}} = c_{\bj_1}\cdots c_{\bj_m}$. Since $\pi$ gives a non-crossing bijection between the two lists $(a_1, \ldots, a_{dm})$ and $(b_1,\ldots, b_{dm})$, it holds that $c_{b_1, \ldots, b_{dm}} = c_{\pi(a_1),\ldots, \pi(a_{dm})}$. Thus, the above sum is
    \begin{align*}
        \ \phi((pp^*)^m) &\le \sum_{\pi \in \NC^*_2(d,m)}  \sum_{a_1, \ldots, a_{dm}} c_{a_1, \ldots, a_{dm}} c_{\pi(a_1),\ldots, \pi(a_{dm})}\\
        \ &\le \sum_{\pi \in \NC^*_2(d,m)} \left(\sum_{a_1, \ldots, a_{dm}} |c_{a_1, \ldots, a_{dm}}|^2\right)^{1/2}
        \left(\sum_{a_1,\ldots, a_{dm}} |c_{\pi(a_1),\ldots, \pi(a_{dm})}|^2\right)^{1/2},
    \end{align*}    
     where the inequality follows from Cauchy-Schwarz. The two internal summations are exactly the same since the summation is over all $dm$ tuples of indices and $\pi$ is a bijection. Switching back to the old indexing scheme, the internal summation then equals 
     \[  \sum_{a_1, \ldots, a_{dm}} |c_{a_1, \ldots, a_{dm}}|^2 = \sum_{\substack{|\bi_1|=\ldots=|\bi_m|=d}} |c_{\bi_1}\cdots c_{\bi_m}|^2 = \left(\sum_{|\bi|=d} |c_{\bi}|^2\right)^m.\] 
    Overall, we have
     \begin{align*}   
        \  \phi((pp^*)^m) &\le |\NC^*_2(d,m)|\left(\sum_{|\bi|=d} |c_{\bi}|^2\right)^m.
    \end{align*}
    Using \lref{lem:catalan-fuss} to bound the number of $*$-pairings,
    \[ |\NC^*_2(d,m)| = C_{d,m} = \frac{1}{m}\binom{m(d+1)}{m-1} = O\left(\frac{(d+1)^{m(d+1)}}{\left(d+\frac1m\right)^{md+1}}\right).\]
    Thus, taking the $m$-th root in the limit $m \to \infty$ yields
    \[ \|p\|^2 = \lim_{m \to \infty}  \phi((pp^*)^m)^{1/m} = \frac{(d+1)^{d+1}}{d^d}\left(\sum_{|\bi|=d} |c_{\bi}|^2\right) \le e(d+1) \left(\sum_{|\bi|=d} |c_{\bi}|^2\right).\]
    This completes the proof of the theorem.
\end{proof}